\documentclass[format=acmsmall, review=false]{acmart}
\usepackage{acm-ec-26}
\usepackage{booktabs} 
\usepackage[ruled]{algorithm2e} 
\usepackage{tcolorbox}
\usepackage{amsmath, amsthm, bm, blkarray, comment, mathtools, enumitem, adjustbox}
\usepackage[subrefformat=parens]{subcaption} 
\usepackage{booktabs}
\usepackage{lipsum}
\usepackage[table]{xcolor}
\usepackage{tcolorbox}

\SetAlFnt{\small}
\SetAlCapFnt{\small}
\SetAlCapNameFnt{\small}
\SetAlCapHSkip{0pt}
\IncMargin{-\parindent}

\newtheorem{assumption}{Assumption}

\setcitestyle{authoryear}

\title[Redefining Stablecoins from Nominal to Real Value: A Maximum Likelihood Approach]{Redefining Stablecoins from Nominal to Real Value: A Maximum Likelihood Approach}

\author{Tomonori Kanno}
\affiliation{%
  \institution{VLUE Inc.}
  \city{Tokyo}
  \country{Japan}
}

\author{Kensuke Ito}
\affiliation{%
  \institution{The University of Tokyo}
  \department{Endowed Chair for Blockchain Innovation}
  \city{Tokyo}
  \country{Japan}
}

\author{Yushi Yoshimura}
\affiliation{%
  \institution{VLUE Inc.}
  \city{Tokyo}
  \country{Japan}
}

\author{Kyohei Shibano}
\affiliation{%
  \institution{The University of Tokyo}
  \department{Endowed Chair for Blockchain Innovation}
  \city{Tokyo}
  \country{Japan}
}

\begin{abstract}
Stablecoins, typically pegged to fiat currencies, cannot achieve true stability because they inherit fluctuations in the underlying unit of account.
To overcome this limitation, we introduce a stablecoin pegged to the \textit{Maximum Likelihood Value} (MLV), a newly defined unit of account derived as the most probable configuration of latent real-value movements that explains observed nominal-value (price) changes.
Grounded in inferential statistics and modern portfolio theory, MLV represents the most stable unit of account, as it enforces a zero real return on the minimum-variance portfolio.
Empirical results confirm the operational viability of an MLV-pegged stablecoin: MLV can be computed in real time from 500 asset price series and improves annualized returns and Sharpe ratios while substantially reducing turnover in portfolio optimization.
\end{abstract}

\begin{document}


\maketitle

\begingroup
\renewcommand{\thefootnote}{}
\footnotetext{%
\textbf{Contact:}
Tomonori Kanno (\texttt{kanno@vlue.jp});
Kensuke Ito (\texttt{k-ito@g.ecc.u-tokyo.ac.jp});
Yushi Yoshimura (\texttt{yoshimura@vlue.jp});
Kyohei Shibano (\texttt{shibano@tmi.t.u-tokyo.ac.jp}).
}
\endgroup

\section*{Disclosure}
The concept of MLV and its application to portfolio optimization are the subject of patent applications filed by VLUE, Inc. in Japan and internationally.

\setcounter{tocdepth}{2} 
\tableofcontents

\section{Introduction} \label{in}

Stablecoins, defined as ``a digital asset that has mechanisms to maintain a low deviation of its price from a target price \cite{woltzenlogel2023stablecoin},'' have supported decentralized systems by enabling low-volatility value storage and cross-border transfer. 
For most stablecoins, achieving low volatility has meant pegging their value to conventional units of account such as fiat currencies and gold.
However, these units of account are not truly stable in terms of value.
As illustrated in Figure \ref{fig:exchange}, the time series of the U.S. dollar (USD) relative to the Japanese yen (JPY) and to gold (XAU) both exhibit substantial volatility, indicating that conventional units of account are themselves volatile.
Consequently, to achieve genuine stability, stablecoins must be redefined as being pegged not to \textit{nominal} values (i.e., prices), but to \textit{real} values, even if such values are unobservable.

\begin{figure}
    \centering
    \begin{minipage}{1.0\linewidth}
        \centering
        \includegraphics[width=0.55\linewidth]{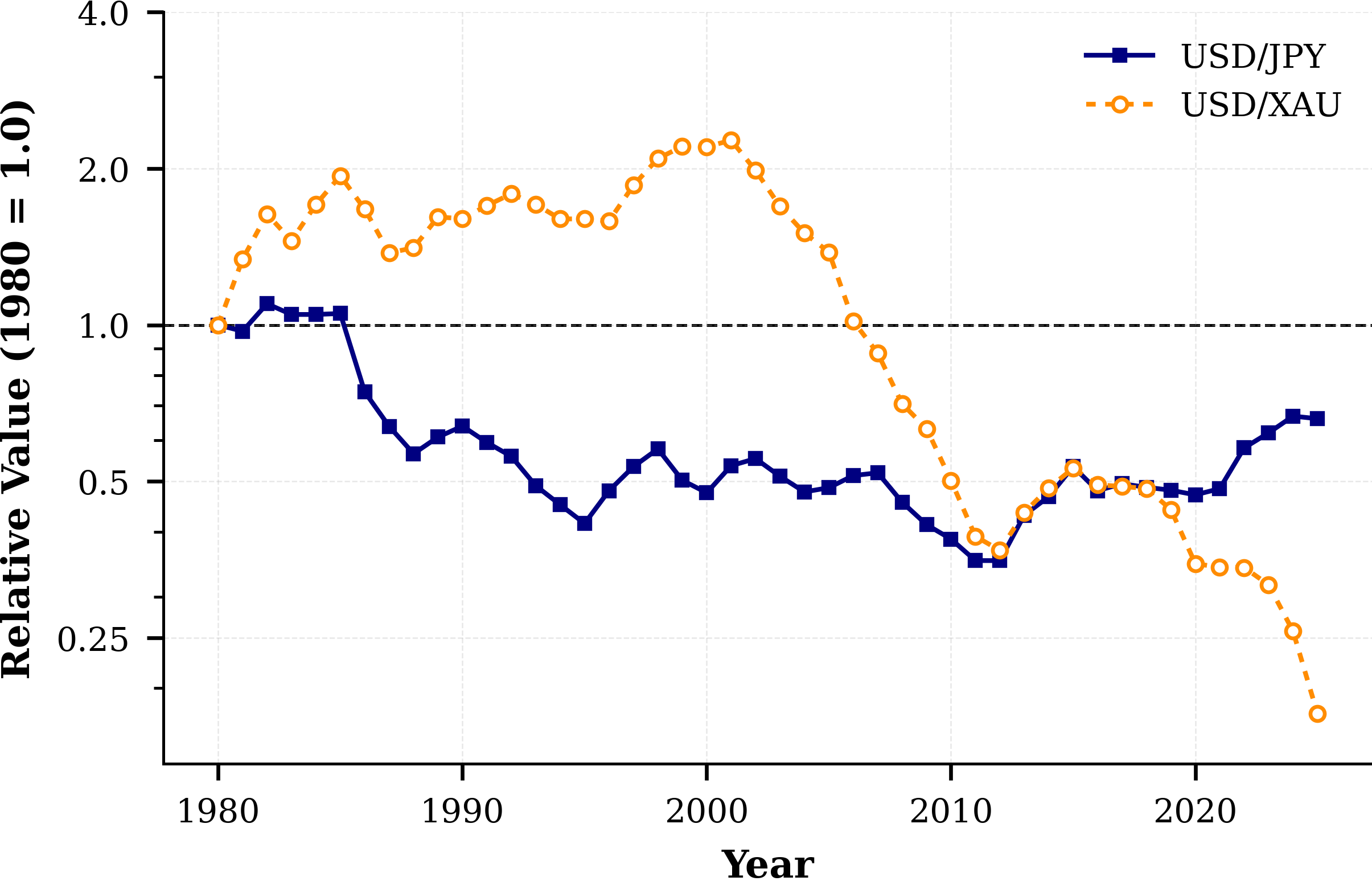}

        \bigskip
        \footnotesize
        \emph{Source:} Author’s visualization based on \citet{imfDataHome} and
        \citet{bankofenglandBankEngland} data.\\
        \emph{Disclosure:} This figure was generated with the assistance of AI and then refined and finalized by the authors.\\
        \emph{Note:} The vertical axis is shown on a log scale.
    \end{minipage}
    \caption{USD/JPY and USD/XAU time series normalized to 1.0 in 1980. The observed volatility illustrates the instability of conventional units of account, highlighting the limitations of existing stablecoin pegs.}
    \label{fig:exchange}
\end{figure}

In economics, real value is often proxied by \textit{purchasing power}, which measures the consumption bundle supported by a unit of currency.
This idea has motivated units of account with constant purchasing power, from the concept of compensated dollar \cite{fisher1913compensated} to recent stablecoin proposals (Section \ref{re}).
However, purchasing power is not well suited to stablecoins used for cross-border value transfer because it relies on region-specific market baskets whose composition varies across jurisdictions and over time.
A real-value-pegged stablecoin therefore requires a new unit of account free from such locality and arbitrariness.

Building on this requirement, we propose a stablecoin pegged to a new unit of account: the \textit{Maximum Likelihood Value} (MLV).
Instead of relying on market baskets, the MLV infers unobservable real value from global financial data.
As shown in Figure \ref{fig:estimation}, it selects the most probable configuration of latent real-value movements (e.g., scenario 3,626: 1.0$\times$ for USD and 0.67$\times$ for JPY) that explains an observed price change (1.5$\times$ for USD/JPY).
Because this inference procedure is uniquely determined and its precision improves with the number of observed assets (Section \ref{ex1}), the MLV avoids the locality and arbitrariness discussed above.
Moreover, the MLV is (i) the most stable unit of account from the perspectives of inferential statistics and \textit{modern portfolio theory} \cite{markowits1952portfolio} (Section \ref{mo-sta}) and (ii) computable in real time from 500 asset price series (Section \ref{ex1}), making it an optimal unit of account for real-value-pegged stablecoins.

\begin{figure}
    \centering
    \includegraphics[width=0.65\linewidth]{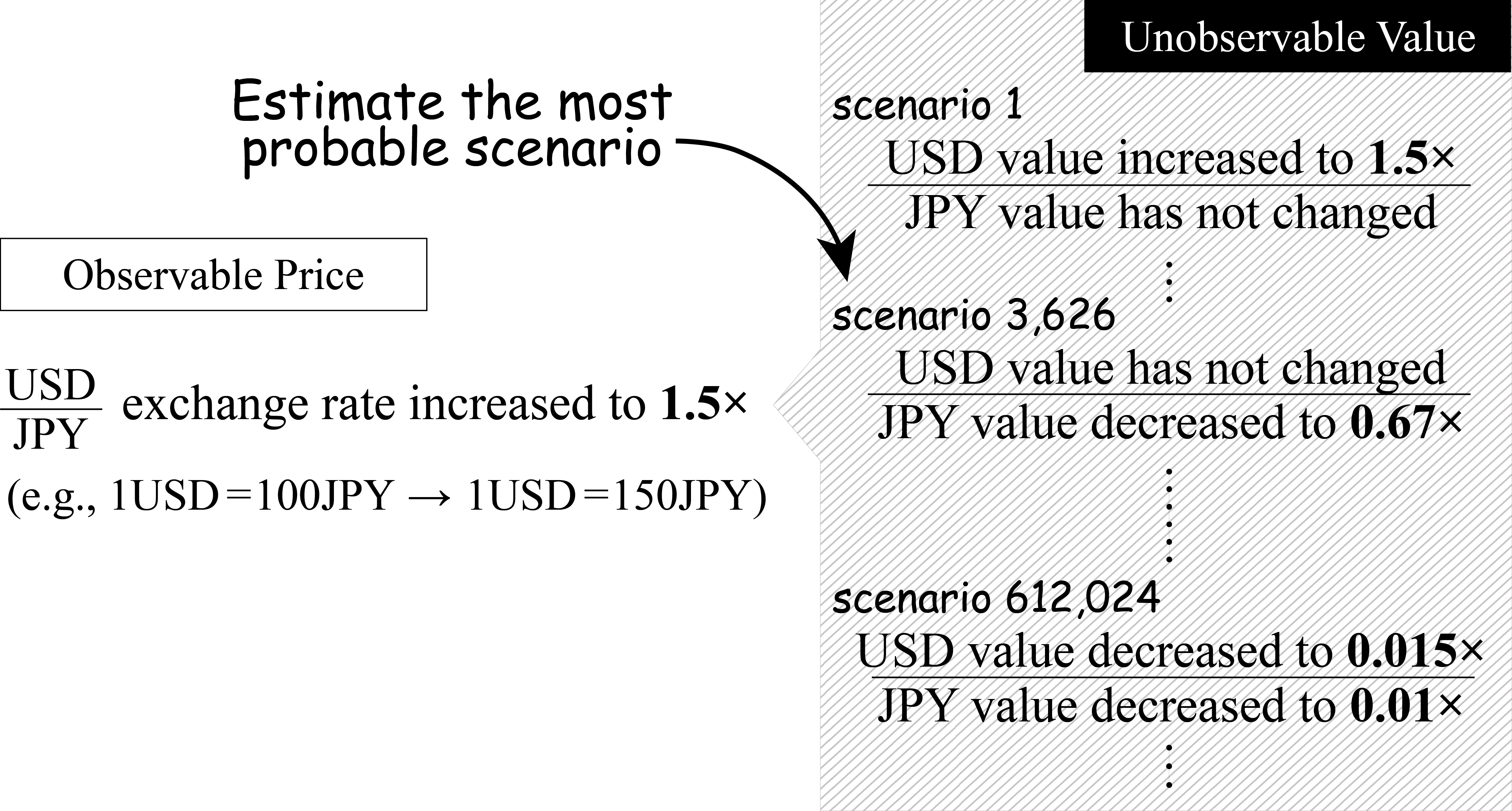}
    \caption{MLV estimation. An observed 1.5$\times$ increase in USD/JPY corresponds to infinitely many possible underlying value-fluctuation scenarios.}
    \label{fig:estimation}
\end{figure}

\begin{figure}
  \centering
  \includegraphics[width=1.0\textwidth]{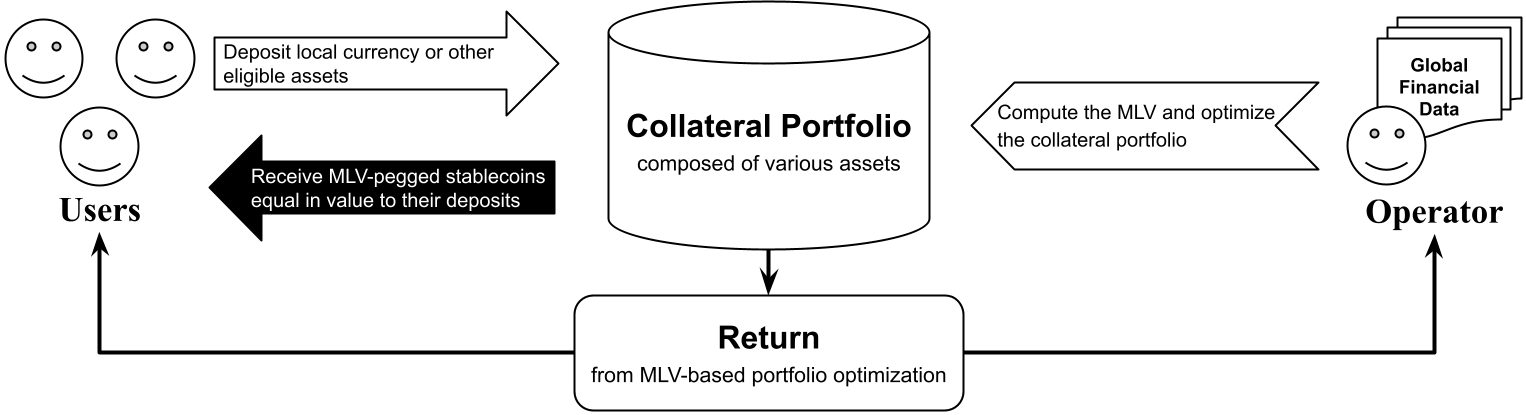}
  \caption{MLV-pegged stablecoins are issued against a collateral portfolio that is periodically rebalanced using the MLV and modern portfolio theory, with the resulting returns distributed among the portfolio, users, and the operator.}
  \label{fig:outview}
\end{figure}

MLV-pegged stablecoins are issued through the scheme illustrated in Figure \ref{fig:outview}, in which \textit{users} deposit assets into a collateral portfolio in exchange for MLV-denominated stablecoins of equivalent value.
In this scheme, an \textit{operator} computes and publishes the MLV using global financial data.
To prevent under-collateralization and incentivize deposits, the operator also performs \textit{MLV-based portfolio optimization}, periodically rebalancing the collateral portfolio using the MLV and modern portfolio theory (Section~\ref{mo-manage})\footnote{When users return MLV-pegged stablecoins, they receive assets of equivalent value in MLV terms from the then-current collateral portfolio. The received assets may therefore differ in composition from their original deposits.}. 
The resulting returns are distributed among the portfolio, users, and the operator.
As demonstrated in Section \ref{ex2}, this optimization improves the portfolio’s annual returns and \textit{Sharpe ratios} \cite{sharpe1964capital} while substantially reducing turnover, thereby providing deposit incentives that are absent in existing fiat-backed stablecoins such as \textit{USDT} \cite{tetherTransparency} and \textit{USDC} \cite{usdcUSDCWorlds}.

\paragraph{Contributions}
In sum, our study makes the following three contributions:
\begin{itemize}
    \item Redefining stablecoins as assets pegged to real rather than nominal values.
    \item Introducing the MLV, an optimal unit of account for real-value-pegged stablecoins. 
    \item Demonstrating the effectiveness of MLV-based portfolio optimization. 
\end{itemize}

\noindent
To the best of our knowledge, this is the first study to estimate unobservable real values and demonstrate the feasibility of a truly stable stablecoin.

The remainder of this paper is organized as follows.
Section~\ref{re} reviews related work,
Section~\ref{mo} presents the MLV framework,
Section~\ref{ex} reports the experimental evaluation, and
Section~\ref{co} concludes with a discussion of future research directions.

\section{Related Work} \label{re}

This section reviews three streams of literature corresponding to the contributions of this study: \textit{stablecoins}, \textit{units of account}, and \textit{modern portfolio theory}.

\subsection{Stablecoins} \label{re-st}

Stablecoins have been surveyed from various perspectives, including collateral type, stabilization mechanism, and underlying blockchain \cite{bullmann2019search, ito2020stablecoin, ante2023systematic, ling2025sok}.
These surveys highlight several attempts to reduce volatility by moving beyond conventional pegs.
For instance, \textit{Libra} \cite{archiveLibraWhite} aimed to enhance stability by pegging its value to a basket of multiple fiat currencies\footnote{On December 1, 2020, Libra was rebranded as \textit{Diem} \cite{diemDiemBlockchain}, marking a shift away from its original goal of basket-based stability. Several studies \cite{zetzsche2019regulating, europaMoneyPrivate, arner2020stablecoins, simmons2021regulating} have summarized the background of this policy change, generally attributing it to conflicts between Libra’s intended function as a payment system and financial regulations in various countries.}, a concept later extended in research on baskets of multiple stablecoins \cite{grobys2025stablecoin}.
Similarly, \textit{Flatcoin} \cite{emmett2023flatcoins} is designed to maintain constant purchasing power by linking its value to an existing inflation index, a concept recently implemented as the \textit{Frax Price Index} \cite{fraxOverviewCPI}.
Our study contributes to this literature by redefining stablecoins as assets pegged to real rather than nominal values through the MLV, a new unit of account that is neither tied to nominal currency baskets nor subject to the locality and arbitrariness of inflation indices.

\subsection{Units of Account} \label{re-ua}

Units of account determine what stablecoins treat as stable, and prior attempts to design more stable units can be grouped into two approaches.
The first is the \textit{basket approach} (as in Libra), which combines existing units of account into a new composite unit.
Notable examples include the \textit{Special Drawing Rights} \cite{imfWhatSDR} and the \textit{European Currency Unit} \cite{louw1987ecu}, the predecessor of the euro, both calculated as weighted averages of multiple fiat currencies.
Although such units are computable in real time, they cannot guarantee real-value stability because they remain exposed to inflation and deflation in their constituent currencies.
For example, \citet{hovanov2004computing} derived minimum-variance weights for a given currency set. 
However, the resulting unit reflects real value only if the aggregate value of the currency set remains constant, i.e., if inflation and deflation across the constituent currencies continuously offset one another.

The second is the \textit{indexed approach} (as in Flatcoin), which pegs a new unit to an existing inflation index that approximates changes in purchasing power.
Examples include the \textit{Unidad de Fomento} \cite{bcentralIndicadoresDiarios} in Chile and the \textit{Unidad de Inversi\'{o}n} \cite{banxicoEstructuraInformacixF3n} in Mexico, both tied to their respective countries' \textit{Consumer Price Index} (CPI)\footnote{These units were developed in Latin American countries between the 1960s and 1990s, a period characterized by unstable price levels. However, \citet{fisher1913compensated} had already emphasized the importance of stable units with constant purchasing power, achievable through indexation. For discussions on the indexed approach in economics, see \citet{shiller1998indexed} and \citet{ho2018search}.}.
Although such units approximate real value through purchasing power, they are neither computable in real time nor free from locality and arbitrariness, because inflation indices are published ex post and depend on region-specific baskets of goods and services.
For instance, \citet{ho2012globalization, ho2018search} derived inflation-adjusted weights for a currency set using national CPIs, effectively creating a hybrid of the basket and indexed approaches.
The resulting unit, however, inherits the same CPI limitations, since the underlying baskets vary across countries and over time.

Thus, at the level of units of account, the same limitations appear: existing approaches are either tied to nominal currency baskets or subject to the locality and arbitrariness of inflation indices.
Our study contributes to this literature by introducing the MLV as an optimal unit of account for real-value-pegged stablecoins, combining real-time computability with high stability while avoiding both limitations.


\subsection{Modern Portfolio Theory}
Modern portfolio theory offers a framework for optimizing portfolios based on the mean-variance tradeoff of asset returns, with the Sharpe ratio providing a measure of risk-adjusted performance.
A key challenge in applying modern portfolio theory is the volatility of the unit of account itself.
For example, a depreciation of the USD can make the returns on portfolio assets appear higher in USD terms, obscuring their real return correlations and complicating optimal allocation.
This divergence between nominal and real returns has been examined since the mid-1970s \cite{adler1975assessment, solnik1978inflation, makin1978portfolio, manaster1979real}, motivating developments such as the \textit{International CAPM} \cite{adler1983international} and empirical \textit{multi-factor models} \cite{chen1986economic, fama1992cross}.
However, these approaches typically incorporate exchange-rate shocks and inflation as additional factors within a nominal framework, leaving unit-of-account volatility itself unaddressed.

Our study contributes to this literature by demonstrating the effectiveness of MLV-based portfolio optimization.
By using the MLV as the benchmark, we implement modern portfolio theory in real-value terms, enabling more efficient portfolio allocation by better capturing underlying real returns and cross-asset correlations.


\section{Model} \label{mo}

\subsection{Definition of the MLV} \label{mo-MLV}

For every pair of assets $i,j\in \{1, \ldots, N\}$\footnote{Let $N$ denote a sufficiently large number.}, let $p_{ij,t}$ denote the \textit{nominal value} (i.e., price) of asset $j$ at time $t$, expressed in units of $i$\footnote{The index $i$ typically represents a common unit of account such as USD and JPY.}.
This nominal value can be decomposed as:

\begin{equation}
    p_{ij,t} = \frac{v_{j,t}}{v_{i,t}},
\end{equation} 
        
\noindent
where $v_{i,t}$ and $v_{j,t}$ denote the \textit{real values} of assets $i$ and $j$ at time $t$.
In this case, the vector $(v_{1,t}, \ldots, v_{N,t})$ cannot be uniquely determined from observable price data, because at any given time, only $N - 1$ independent price points are available for $N$ assets\footnote{For example, when $N = 3$ assets $\{\mathrm{USD}, \mathrm{JPY}, \mathrm{XAU}\}$ are considered, three price pairs can be observed (e.g., USD/JPY, JPY/XAU, and XAU/USD). Since any one of these can be uniquely determined from the other two, the number of independent price data points is two.}, i.e., one degree of freedom is always missing.
To resolve this indeterminacy, the MLV estimates real values by analyzing the stochastic behavior of both nominal and real values through a maximum likelihood framework.

To streamline notation, let $\Delta p_{ij,t}\coloneqq \ln \left(p_{ij,t}/p_{ij,t-1} \right)$, $\Delta v_{i,t}\coloneqq \ln \left( v_{i,t}/v_{i,t-1} \right)$, and $\mathbf{\Delta v}_{t} \coloneqq (\Delta v_{1,t}, \ldots, \\ \Delta v_{N,t})^{\top}$.


\begin{definition}[The MLV] \label{defMLV}
The MLV of asset $i$ at time $t$ is the value of $\Delta v_{i,t}$ that maximizes the joint probability density of $\mathbf{\Delta v}_{t}$.
\end{definition}

\noindent
In other words, the MLV provides an estimate of each asset's real log return\footnote{\citet{hovanov2004computing} can be interpreted as a model that computes $\mathbf{\Delta v}_{t}$ by imposing the constraint $\sum_{i=1}^{N}\Delta v_{i,t} = 0$ for all $t$, thereby eliminating the remaining degree of freedom.}.
For clarity, we denote this maximizer by $\Delta \widehat{v}_{i,t}$ and collect these estimates in $\mathbf{\Delta \widehat{v}}_{t} \coloneq (\Delta \widehat{v}_{1,t}, \ldots, \Delta \widehat{v}_{N,t})^{\top}$.
If the sequence $\{\mathbf{\Delta \widehat{v}}_{t}\}_{t=1}^{T}$ is available, then a volatility-free unit of account can be constructed over $t=0, \ldots, T$.

\begin{definition}[The MLV as a Unit of Account] \label{defMLVsta}
The MLV as a unit of account fixes $v_{i,0} = 1$ and measures the real value of all assets by using the sequence $\{\mathbf{\Delta \widehat{v}}_{t}\}_{t=1}^{T}$.
\end{definition}

\noindent
A specific example is provided below.
Note that the choice of reference asset $i$ is immaterial, since it only scales all MLV-denominated values by a constant factor without affecting their relative ratios.

\begin{center}
\begin{tcolorbox}[
  width=0.85\linewidth,
  title={\em Example: The MLV as a Unit of Account}
]

We illustrate the concept by using the following observable price data:

\vspace{0.5\baselineskip}
{
\begin{tabular}{llll}

\toprule
                 & \textbf{$t=0$} & \textbf{$t=1$} & \textbf{$t=2$} \\ \midrule
$p_{\mathrm{JPYUSD}}$ & 150          & 100          & 200          \\ 
$p_{\mathrm{XAUJPY}}$ & 600,00$0^{-1}$        & 500,00$0^{-1}$          & 1,200,00$0^{-1}$         \\ 
$p_{\mathrm{USDXAU}}$ & 4,000         & 5,000         & 6,000        \\ \bottomrule
\end{tabular}
}
\vspace{0.7\baselineskip}

\noindent 
Assume that we estimate the following MLVs:

\vspace{0.5\baselineskip}
{
\begin{tabular}{llll}

\toprule
                 & \textbf{$t=0$} & \textbf{$t=1$} & \textbf{$t=2$} \\ \midrule
$\Delta \widehat{v}_{\mathrm{USD}}$ & -         & $-0.4$          & $+0.7$          \\ 
$\Delta \widehat{v}_{\mathrm{JPY}}$ & -        & $+0.1$          & $-0.9$         \\ 
$\Delta \widehat{v}_{\mathrm{XAU}}$ & -         & $+0.3$         & $+0.2$        \\ \bottomrule
\end{tabular}
}

\vspace{0.7\baselineskip}

\noindent    
If we set $v_{\mathrm{USD},0} = 1$, then the real values of all assets are measured as follows:

\vspace{0.5\baselineskip}
{
\begin{tabular}{llll}

\toprule
                 & \textbf{$t=0$} & \textbf{$t=1$} & \textbf{$t=2$} \\ \midrule
$v_{\mathrm{USD}}$ & $1$          & $1\cdot e^{-0.4}$          & $1\cdot e^{-0.4+0.7}$          \\ 
$v_{\mathrm{JPY}}$ & $150^{-1}$        & $150^{-1}\cdot e^{0.1}$          & $150^{-1}\cdot e^{0.1-0.9}$         \\ 
$v_{\mathrm{XAU}}$ & 4,000         & 4,00$0\cdot e^{0.3}$         & 4,00$0\cdot e^{0.3+0.2}$        \\ \bottomrule
\end{tabular}
}
\end{tcolorbox}
    
\end{center}

\subsection{Derivation of the MLV} \label{de-MLV}

$\Delta \widehat{v}_{i,t}$ can be derived by imposing the following assumption on the stochastic behavior of $(v_{1,t}, \ldots, v_{N,t})$\footnote{Strictly speaking, the derivation requires two additional assumptions: the no-arbitrage condition and the law of one price.}.

\begin{assumption}[Elliptical Distribution] \label{as-Ell}
The distribution of $\mathbf{\Delta v}_{t}$ belongs to the family of elliptical distributions\footnote{The family of elliptical distributions includes the multivariate normal distribution as a special case, as well as many non-normal multivariate distributions such as the multivariate Cauchy, multivariate exponential, and multivariate $t$ distributions.}.
\end{assumption}

\noindent
Elliptical distributions have the following form for their joint probability density, which is maximized when $\mathbf{x}^{\top} \Sigma^{-1} \mathbf{x}$ is minimized:

\begin{equation}
\label{eqn:ellip}
    f(\mathbf{x}, \Sigma) = k\cdot g \left( \mathbf{x}^{\top} \Sigma^{-1} \mathbf{x} \right),
\end{equation}

\noindent
where $\mathbf{x}$ is a random variable vector, 
$\Sigma$ is a positive semi-definite matrix characterizing the covariance structure of $\mathbf{x}$, 
$k$ is a normalization constant, 
and $g (\cdot)$ is a generator function. 

Substituting the corresponding variables into \eqref{eqn:ellip} yields:

\begin{equation}
\label{eqn:dist}
\begin{aligned}
    f(\mathbf{\Delta v}_{t}, \Sigma_{t}) &= f(\mathbf{\Delta p}_{i,t} + \Delta v_{i,t}\mathbf{1}, \Sigma_{t}) \\
     &= k\cdot g \left( (\mathbf{\Delta p}_{i,t} + \Delta v_{i,t}\mathbf{1})^{\top} \Sigma_{t}^{-1} (\mathbf{\Delta p}_{i,t} + \Delta v_{i,t}\mathbf{1}) \right),
\end{aligned}
\end{equation}

\noindent
where $\Sigma_{t}$ is an $N \times N$ matrix representing the covariance structure of $\mathbf{\Delta v}_{t}$\footnote{Note that $\Sigma_t$ is not restricted to the conventional covariance matrix. Instead, this derivation adopts \textit{Tyler’s scatter matrix} \cite{tyler1987distribution}, a generalization of the covariance matrix that is well suited to elliptical distributions of asset returns, particularly those exhibiting heavy-tailed behavior (see Appendix~\ref{app-sigma} for details). In contrast, the experiments in Section \ref{ex1} adopt the correlation matrix as a representation of $\Sigma_t$ that abstracts from scale information.}, $\mathbf{\Delta p}_{i,t} \coloneq (\Delta p_{i1,t}, \ldots,\\ \Delta p_{iN,t})^{\top}$ denotes the vector of log-return prices of all assets at time $t$ in units of $i$, and $\mathbf{1} \coloneq (1, \ldots,1)^{\top}$ is the $N$-dimensional all-one vector.
 

To find $\Delta v_{i,t}$ that maximizes \eqref{eqn:dist}, we solve the following optimization problem:

\begin{equation}
\label{eqn:opt}
    \min_{\Delta v_{i,t}} \ (\mathbf{\Delta p}_{i,t} + \Delta v_{i,t}\mathbf{1})^{\top} \Sigma_{t}^{-1} (\mathbf{\Delta p}_{i,t} + \Delta v_{i,t}\mathbf{1}).
\end{equation}


\noindent
The analytical solution of \eqref{eqn:opt} is:

\begin{equation}
\label{eqn:mlv}
    \Delta \widehat{v}_{i,t} = -\frac{\mathbf{\Delta p}_{i,t}^\top \Sigma_{t}^{-1} \mathbf{1}}{\mathbf{1}^\top \Sigma_{t}^{-1} \mathbf{1}}.
\end{equation}    

\noindent
Thus, $\Delta \widehat{v}_{i,t}$ is solely determined by $\mathbf{\Delta p}_{i,t}$ and $\Sigma_{t}$.
Since \eqref{eqn:mlv} enables an iterative procedure that estimates $\Sigma_t$ from past observable price log-returns (see Appendix~\ref{app-sigma} for details), $\Delta \widehat{v}_{i,t}$ can be entirely computed from observable price data.

Once $\Delta \widehat{v}_{i,t}$ is obtained, the vector $\mathbf{\Delta \widehat{v}}_{t}$ is given by $\mathbf{\Delta \widehat{v}}_{t} = \mathbf{\Delta p}_{i,t} + \Delta \widehat{v}_{i,t}\mathbf{1}$.
The full sequence $\{\mathbf{\Delta \widehat{v}}_{t}\}_{t=1}^{T}$ is then constructed by applying \eqref{eqn:mlv} for $t=1, \ldots, T$.





\subsection{Stability of the MLV} \label{mo-sta}
The stability of the MLV is not only supported by inferential statistics, but also by modern portfolio theory.
Consider a portfolio composed of $N$ distinct assets. 
Under the mean-variance framework of modern portfolio theory, the minimum-variance portfolio is defined as the solution to the following optimization problem, subject to the full-investment constraint:

\begin{equation}
\label{eqn:mvp}
    \min_{\mathbf{w}} \ \mathbf{w}^\top \Sigma \mathbf{w} 
    \quad \text{s.t.} \quad \mathbf{w}^\top \mathbf{1} = 1,
\end{equation}

\noindent
where $\mathbf{w}$ is the $N$-dimensional weight vector of assets, and $\Sigma$ denotes the covariance structure of asset returns.
The analytical solution of \eqref{eqn:mvp} is:

\begin{equation}
\label{eqn:mvpw}
    \mathbf{w}^* = \frac{\Sigma^{-1} \mathbf{1}}{\mathbf{1}^\top \Sigma^{-1} \mathbf{1}}.
\end{equation}

\noindent
This portfolio achieves the lowest possible risk among all the fully invested portfolios, regardless of expected returns.

Using \eqref{eqn:mvpw}, we can derive the following theorem regarding the stability of the MLV.

\begin{theorem} \label{theoStaMLV}
The MLV is the unit of account that ensures the real return of the minimum-variance portfolio is zero.
\end{theorem}

\begin{proof}
Expanding \eqref{eqn:mlv} with respect to $\mathbf{\Delta p}_{i,t}$ gives:
\begin{equation}
    \Delta \widehat{v}_{i,t} = -\frac{(\mathbf{\Delta \widehat{v}}_{t} - \Delta \widehat{v}_{i,t} \mathbf{1})^\top \Sigma_{t}^{-1} \mathbf{1}}{\mathbf{1}^\top \Sigma_{t}^{-1} \mathbf{1}},
\end{equation}

\noindent
which, when rearranged, yields:
\begin{equation}
\label{eqn:solv}
    \frac{\mathbf{\Delta \widehat{v}}_{t}^\top \Sigma_{t}^{-1} \mathbf{1}}{\mathbf{1}^\top \Sigma_{t}^{-1} \mathbf{1}} = 0.
\end{equation}

\noindent
From \eqref{eqn:mvpw}, \eqref{eqn:solv} can be rewritten as follows:
\begin{equation}
    \mathbf{\Delta \widehat{v}}_{t}^\top \mathbf{w}_t^* = 0,
\end{equation}

\noindent
where the left-hand side represents the MLV-denominated return of the minimum-variance portfolio, which equals zero for any $t$. 
\end{proof}

\noindent
Given the mean-variance tradeoff of asset returns, this finding implies that the MLV is the most stable unit of account under modern portfolio theory.
In other words, estimating unobservable real-value fluctuations via maximum likelihood is equivalent to defining them such that the real return of the minimum-variance portfolio is always zero.

\subsection{MLV-based Portfolio Optimization} \label{mo-manage}

It is worth noting that the MLV enhances portfolio optimization because $\Sigma_{t}$, defined for $\mathbf{\Delta v}_{t}$ rather than $\mathbf{\Delta p}_{i,t}$, accurately captures the covariance structure of asset returns and thereby enables a faithful implementation of modern portfolio theory.

In practice, portfolio rebalancing at time $t$ is formulated as the following optimization problem:

\begin{equation}
\label{eqn:mlvpo}
\begin{aligned}
\max_{\mathbf{w}_t}\quad
& \widehat{\boldsymbol{\mu}}_{i,t}^\top \mathbf{w}_t
  - C \lVert \mathbf{w}_t - \mathbf{w}_{t-1} \rVert_{1} \\
\text{s.t.}\quad
& \mathbf{w}_t^\top \mathbf{1} = 1, \\
& \mathbf{w}_t^\top \Sigma_t \mathbf{w}_t \le \theta .
\end{aligned}
\end{equation}

\noindent
where $\widehat{\boldsymbol{\mu}}_{i, t}$ denotes the estimated expected dividend return vector for each asset at time $t$ in units of $i$ (see Appendix \ref{app-mu} for the formal definition)\footnote{Since portfolio weights $\mathbf{w}_t$ for the interval $(t, t+1\rbrack$ are chosen \textit{ex ante} at time $t$, portfolio optimization in practice is typically formulated as the maximization of $\widehat{\boldsymbol{\mu}}_{i, t}^\top \mathbf{w}_t$, rather than $\mathbf{\Delta p}_{i,t}^\top \mathbf{w}_t$ or $\mathbf{\Delta \widehat{v}}_{t}^\top \mathbf{w}_t$, due to the greater reliability of estimating expected returns \cite{shiller1981stock}.},
$C$ is the rebalancing cost coefficient proportional to the $L^1$ norm of $\mathbf{w}_t - \mathbf{w}_{t-1}$, and
$\theta$ is the risk-tolerance threshold\footnote{In practice, the risk constraint can be formulated using alternative downside risk measures, such as CVaR (Section \ref{ex2}).}.
The performance of this MLV-based portfolio optimization is evaluated empirically in Section \ref{ex2}.


\section{Experiments} \label{ex}

MLV-pegged stablecoins, issued according to the scheme in Figure \ref{fig:outview}, are feasible if (i) the MLV can be computed in real time and (ii) MLV-based portfolio optimization can generate sufficient returns for users.
We evaluated both of these conditions through a series of experiments.

\subsection{Real-time Computability of the MLV} \label{ex1}

The first experiment evaluates the real-time computability of the MLV by examining how many observable assets are required to accurately estimate it in a rolling, online setting.

\paragraph{Data}
Multivariate $t$-distributed log-return time series of length $T=10N$ are generated by using a random $N \times N$ correlation matrix. 
These series represent the unobservable real log-returns for $N$ assets, serving as the ground truth.
The corresponding observable price log-returns, $\{\mathbf{\Delta p}_{i,t}\}_{t=1}^{10N}$, are then constructed by fixing the first asset as the unit of account (i.e., $i=1$) and subtracting its series from those of all other assets at each time step.
In this experiment, $N$ is set to $500$.

\paragraph{Methods}
For each run, $n$ assets are randomly selected once from the universe of $N$.
Using the corresponding $n$-dimensional price log-returns extracted from $\mathbf{\Delta p}_{1,t}$, the MLV vector $\mathbf{\Delta \widehat{v}}_{t}$ and its correlation matrix $\Sigma_t$\footnote{For simplicity, this experiment represents $\Sigma_t$ using correlation matrices, thereby abstracting from scale information. Correlations are computed using Kendall's $\tau$–based rank-correlation estimator to account for the heavy-tailed nature of asset returns.} are estimated.
Estimation is performed at the end of a rolling window of length $3N$, yielding $7N+1$ correlation-matrix estimates $\{\Sigma_t\}_{t=3N}^{10N}$ per run.
As ground truth, another set of $7N+1$ correlation-matrix estimates $\{\Sigma_t^{\mathrm{true}}\}_{t=3N}^{10N}$ is constructed per run, where each $\Sigma_t^{\mathrm{true}}$ is computed from the real log-returns of all $N$ assets over a rolling window ending at time $t$ (i.e., $3N^2$ data points).
Accuracy is evaluated using the average \textit{Root-Mean-Square Error} (RMSE) between $\Sigma_t$ and $\Sigma_t^{\mathrm{true}}$ across the $7N+1$ estimates.
This procedure is repeated 10 times for each $n \in \{50,100,150,\ldots,500\}$.

\paragraph{Results}
As shown in Figure \ref{fig:rmse_corr_vs_assets_mlv}, the estimation error rapidly decreases as the number of observed assets increases and stabilizes beyond approximately $n=200$. 
This indicates diminishing accuracy gains with additional observations.
From a practical perspective, observing on the order of $n\approx 500$ assets is sufficient. 
Given that the dominant computational cost of each rolling update scales as $O(n^3)$, this observation level allows MLV to be computed in real time on standard computing hardware.

\begin{figure}
  \centering
  \includegraphics[width=0.6\linewidth]{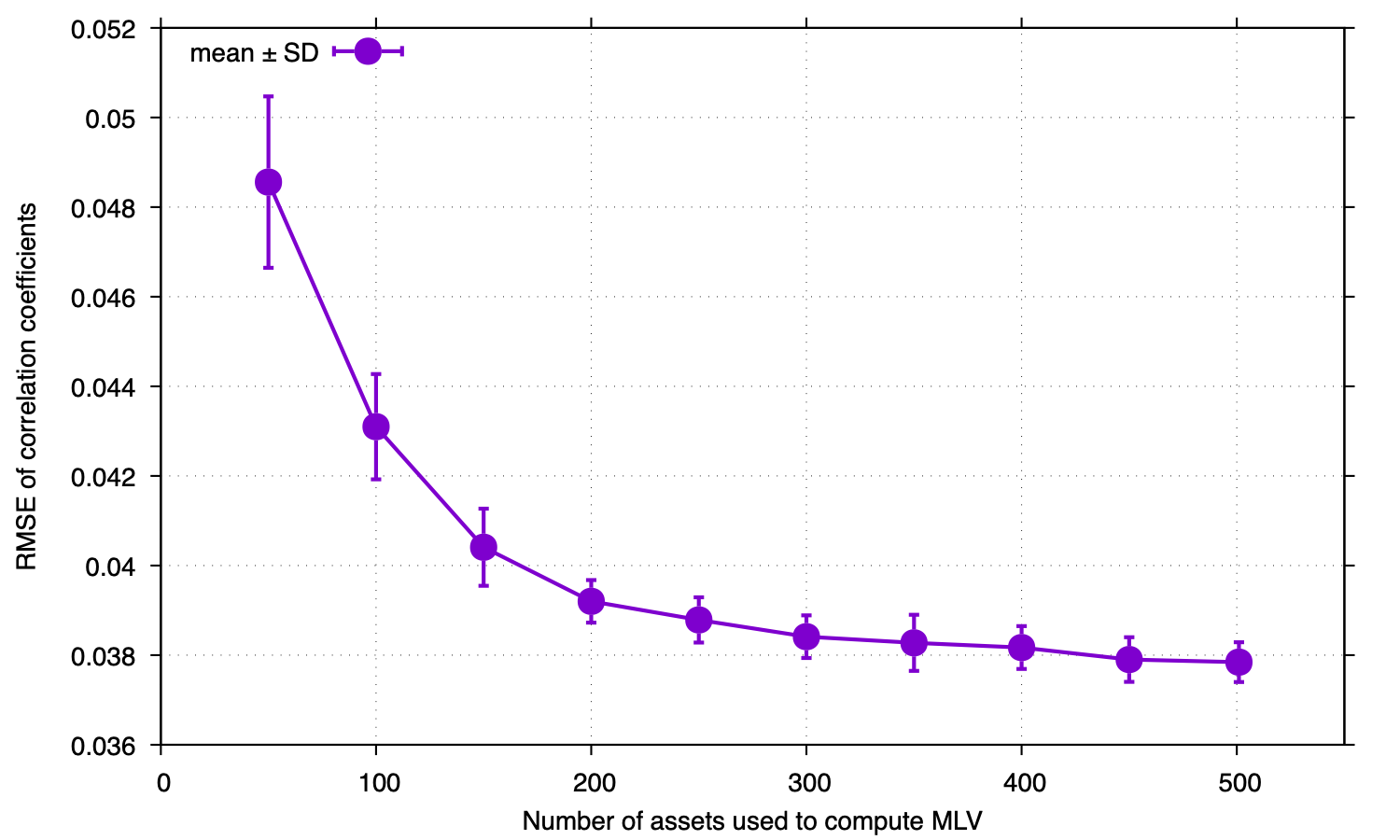}
  \caption{
  Mean RMSE (with standard error) between the MLV-based and the ground-truth correlation matrices, shown as a function of the number of observed assets $n$. RMSE is averaged over the rolling windows per run, and the mean and standard error are computed across 10 independent runs.
  }
  \label{fig:rmse_corr_vs_assets_mlv}
\end{figure}

\subsection{Performance of MLV-based Portfolio Optimization} \label{ex2} 

The second experiment evaluates the performance of MLV-based portfolio optimization by comparing the empirical outcome of \eqref{eqn:mlvpo} across USD- and MLV-based risk and evaluation spaces.

\paragraph{Data}
U.S.\ equities from the S\&P~500 universe are used for this experiment.
At each rebalancing date, set at five week intervals, the collateral portfolio is constructed by selecting high-dividend stocks based on their trailing 12-month dividend yields.
A broader set of equities ($n \approx 500$) from the same S\&P~500 is used to estimate the MLV time series.
Dividend payments are treated as nominal cash flows and aggregated at the portfolio level.
The sample period covers June 2019 through January 2026.

\paragraph{Methods}
At each rebalancing date $t$, portfolio weights $\mathbf{w}_t$ are determined by solving the dividend-maximization problem \eqref{eqn:mlvpo} under an explicit downside risk constraint, with the rebalancing cost coefficient set to $C = 0.01$ and dividend returns expressed in USD, i.e., $\widehat{\boldsymbol{\mu}}_{\mathrm{USD},t}$.
For practicality, downside risk is controlled not by the variance constraint $\theta$ but by a weekly \textit{Conditional Value at Risk} (CVaR) constraint, $\mathrm{CVaR}_{0.95} \ge -0.04$\footnote{Specifically, portfolio returns (asset price changes and dividend payments) over the past 100 weeks are ranked, and the average of the worst 5\% outcomes is constrained to be no less than $-4\%$.}.
This experiment compares two risk spaces by computing CVaR with covariance structures derived from either USD-based or MLV-based returns, while all other parameters are held constant.
For each risk space, portfolio performance is evaluated in both USD and MLV spaces, where all metrics (annualized return, annualized volatility, Sharpe ratio, maximum drawdown, and portfolio turnover) are computed in the corresponding unit of account\footnote{Since dividend returns are expressed in USD, i.e., $\widehat{\boldsymbol{\mu}}_{\mathrm{USD},t}$, the difference in evaluation space reflects whether the same portfolio optimization outcome is evaluated in USD or in MLV.}.

\begin{table}[t]
\centering
\caption{Performance of dividend-maximizing portfolios under CVaR constraints}
\label{tab:divmax_cvar}
\small
\begin{subtable}[t]{\textwidth}
\centering
\caption{Evaluated in the USD space}
\label{tab:divmax_cvar_a}
\begin{tabular}{lccccc}
\toprule
Risk Space & Ann.\ Return & Ann.\ Volatility & Sharpe Ratio & Max Drawdown & Turnover \\
\midrule
USD & 0.117 & 0.211 & 0.554 & -0.376 & 0.529 \\
\rowcolor{gray!10}
MLV   & 0.156 & 0.223 & 0.699 & -0.466 & 0.133 \\
\bottomrule
\end{tabular}
\end{subtable}

\vspace{1.0\baselineskip}

\begin{subtable}[t]{\textwidth}
\centering
\caption{Evaluated in the MLV space}
\label{tab:divmax_cvar_b}
\begin{tabular}{lccccc}
\toprule
Risk Space & Ann.\ Return & Ann.\ Volatility & Sharpe Ratio & Max Drawdown & Turnover \\
\midrule
USD & 0.037 & 0.132 & 0.281 & -0.220 & 0.527 \\
\rowcolor{gray!10}
MLV   & 0.053 & 0.116 & 0.461 & -0.148 & 0.137 \\
\bottomrule
\end{tabular}
\end{subtable}

\end{table}

\paragraph{Results}
Table \ref{tab:divmax_cvar} reports the performance of portfolio optimization across different combinations of risk and evaluation spaces.
In the USD evaluation, using the MLV as the risk space increases the annualized return from 11.7\% to 15.6\% and the Sharpe ratio from 0.554 to 0.699, while substantially reducing portfolio turnover from 0.529 to 0.133.
In the MLV evaluation, a comparable reduction in turnover is observed, while the percentage increases in annualized return and Sharpe ratio are even larger than those under USD evaluation\footnote{The annualized returns in the MLV evaluation (3.7\% and 5.3\%) are lower than those in the USD evaluation (11.7\% and 15.6\%). This difference likely reflects the depreciation of the USD over the sample period from June 2019 to January 2026.}.
Overall, these results demonstrate that MLV-based portfolio optimization improves return performance and rebalancing stability under CVaR constraints, relative to USD-based portfolio optimization.

\section{Conclusion} \label{co}

This study proposed stablecoins pegged to the MLV, a new unit of account that infers unobservable real-value movements from financial asset prices.
We first theoretically established that the MLV is the most stable unit of account from the perspectives of both inferential statistics and modern portfolio theory (Section \ref{mo}).
We then experimentally demonstrated that the MLV can be computed in real time and can improve portfolio optimization performance (Section \ref{ex}).
These results indicate that stablecoins can be redefined as truly stable assets, i.e., assets pegged to real rather than nominal values.
Finally, we identify several promising directions for future work.

\paragraph{Refining the Derivation of the MLV}
The derivation of the MLV (Section \ref{de-MLV}) relies on the stochastic behavior of $(v_{1,t}, \ldots, v_{N,t})$ and the accuracy of $\Sigma_t$. 
Assumption \ref{as-Ell} regarding the stochastic distribution is not overly restrictive, since many commonly used distributions belong to the elliptical family, although further relaxation may be possible.
It is also important to develop more accurate and real-time methods for estimating $\Sigma_t$ than the approach described in Appendix \ref{app-sigma}\footnote{For example, \citet{hayashi2005covariance} proposed a covariance estimator for non-synchronously observed high-frequency data, addressing biases that arise when asset prices are observed at different trading times.}. 
Such refinements can enhance the MLV's ability to capture stylized features of financial markets, including \textit{tail dependence} \cite{hauksson2001multivariate} and \textit{asymmetric dependence} \cite{longin2001extreme}\footnote{In risk evaluation, \textit{copulas} \cite{nelsen2006introduction} are frequently used to model these stylized features. However, in the derivation of the MLV, where $\Sigma_t$ explicitly appears in the analytical solution (5), pursuing accurate and real-time updates of $\Sigma_t$ constitutes a more consistent approach.}.

\paragraph{Making MLV-based Portfolio Optimization More Practical}
Several practical challenges remain in implementing MLV-based portfolio optimization (Section~\ref{mo-manage}).
One challenge involves price data acquisition and portfolio rebalancing under region-specific trading hours and time zones, which may draw on insights from existing studies of asynchronous trading and non-overlapping market hours \cite{eun1989international, baumohl2011stock}.
Another challenge concerns the optimal allocation of realized returns among the portfolio, users, and the operator; addressing this challenge may draw on existing studies of fee and compensation structures in mutual and hedge funds that analyze how performance is shared among stakeholders \cite{elton2003incentive, escobar2018optimal, ben2020performance}.

\paragraph{Decentralizing the Issuance of MLV-pegged Stablecoins}
MLV-pegged stablecoins may be issued (Figure \ref{fig:outview}) without the operator. 
Such decentralization enhances the robustness of the issuance scheme by eliminating single points of failure, while aligning with the design principles of the \textit{Bitcoin protocol} \cite{nakamoto2008bitcoin}.
Relevant precedents include the \textit{DAI stablecoin} \cite{dai2022maker}, which mitigates under-collateralization and incentivizes deposits without a centralized operator, as well as lessons learned from failed \textit{algorithmic stablecoins} \cite{al2017basis, kereiakes2019terra}.
More broadly, sustainable decentralization requires incentive mechanisms that are grounded in \textit{cryptoeconomics} and \textit{tokenomics} \cite{kensuke2024cryptoeconomics}.



\section*{Acknowledgment}
The authors gratefully acknowledge the \textit{Endowed Chair for Blockchain Innovation} and the \textit{Mohammed bin Salman Center for Future Science and Technology for Saudi-Japan Vision 2030} (MbSC2030) for their financial support.

\section*{Author Contribution}
T.K. conceived the main conceptual framework and developed the mathematical model. 
K.I. drafted the manuscript and contributed to the overall research design. 
Y.Y. assisted in model development. 
K.S. provided critical revisions. 

\bibliographystyle{ACM-Reference-Format}
\bibliography{sample-bibliography}


\begin{thebibliography}{52}


\ifx \showCODEN    \undefined \def \showCODEN     #1{\unskip}     \fi
\ifx \showISBNx    \undefined \def \showISBNx     #1{\unskip}     \fi
\ifx \showISBNxiii \undefined \def \showISBNxiii  #1{\unskip}     \fi
\ifx \showISSN     \undefined \def \showISSN      #1{\unskip}     \fi
\ifx \showLCCN     \undefined \def \showLCCN      #1{\unskip}     \fi
\ifx \shownote     \undefined \def \shownote      #1{#1}          \fi
\ifx \showarticletitle \undefined \def \showarticletitle #1{#1}   \fi
\ifx \showURL      \undefined \def \showURL       {\relax}        \fi
\providecommand\bibfield[2]{#2}
\providecommand\bibinfo[2]{#2}
\providecommand\natexlab[1]{#1}
\providecommand\showeprint[2][]{arXiv:#2}

\bibitem[Adler and Biger(1975)]%
        {adler1975assessment}
\bibfield{author}{\bibinfo{person}{Michael Adler} {and} \bibinfo{person}{Nahum
  Biger}.} \bibinfo{year}{1975}\natexlab{}.
\newblock \showarticletitle{The assessment of inflation and portfolio
  selection}.
\newblock \bibinfo{journal}{\emph{the Journal of Finance}}
  \bibinfo{volume}{30}, \bibinfo{number}{2} (\bibinfo{year}{1975}),
  \bibinfo{pages}{451--467}.
\newblock


\bibitem[Adler and Dumas(1983)]%
        {adler1983international}
\bibfield{author}{\bibinfo{person}{Michael Adler} {and}
  \bibinfo{person}{Bernard Dumas}.} \bibinfo{year}{1983}\natexlab{}.
\newblock \showarticletitle{International portfolio choice and corporation
  finance: A synthesis}.
\newblock \bibinfo{journal}{\emph{The journal of finance}}
  \bibinfo{volume}{38}, \bibinfo{number}{3} (\bibinfo{year}{1983}),
  \bibinfo{pages}{925--984}.
\newblock


\bibitem[Al-Naji et~al\mbox{.}(2017)]%
        {al2017basis}
\bibfield{author}{\bibinfo{person}{Nader Al-Naji}, \bibinfo{person}{Josh Chen},
  {and} \bibinfo{person}{Lawrence Diao}.} \bibinfo{year}{2017}\natexlab{}.
\newblock \showarticletitle{Basis: a price-stable cryptocurrency with an
  algorithmic central bank}.
\newblock \bibinfo{journal}{\emph{Basis. io}} (\bibinfo{year}{2017}).
\newblock


\bibitem[Ante et~al\mbox{.}(2023)]%
        {ante2023systematic}
\bibfield{author}{\bibinfo{person}{Lennart Ante}, \bibinfo{person}{Ingo
  Fiedler}, \bibinfo{person}{Jan~Marius Willruth}, {and} \bibinfo{person}{Fred
  Steinmetz}.} \bibinfo{year}{2023}\natexlab{}.
\newblock \showarticletitle{A systematic literature review of empirical
  research on stablecoins}.
\newblock \bibinfo{journal}{\emph{FinTech}} \bibinfo{volume}{2},
  \bibinfo{number}{1} (\bibinfo{year}{2023}), \bibinfo{pages}{34--47}.
\newblock


\bibitem[Arner et~al\mbox{.}(2020)]%
        {arner2020stablecoins}
\bibfield{author}{\bibinfo{person}{DW Arner}, \bibinfo{person}{R Auer}, {and}
  \bibinfo{person}{J Frost}.} \bibinfo{year}{2020}\natexlab{}.
\newblock \showarticletitle{Stablecoins: risks, potential and regulation (BIS
  Working Papers No. 905)}.
\newblock \bibinfo{journal}{\emph{Bank for International Settlements}}
  (\bibinfo{year}{2020}).
\newblock


\bibitem[{Banco Central de Chile}(2025)]%
        {bcentralIndicadoresDiarios}
\bibfield{author}{\bibinfo{person}{{Banco Central de Chile}}.}
  \bibinfo{year}{2025}\natexlab{}.
\newblock \bibinfo{title}{{I}ndicadores diarios --- si3.bcentral.cl}.
\newblock
  \bibinfo{howpublished}{\url{https://si3.bcentral.cl/indicadoressiete/secure/IndicadoresDiarios.aspx}}.
\newblock
\newblock
\shownote{[Accessed 01-10-2025]}.


\bibitem[{Banco de M\'{e}xico}(2025)]%
        {banxicoEstructuraInformacixF3n}
\bibfield{author}{\bibinfo{person}{{Banco de M\'{e}xico}}.}
  \bibinfo{year}{2025}\natexlab{}.
\newblock \bibinfo{title}{Estructura de informaci\'{o}n (SIE, Banco de
  {M}\'{e}xico) --- banxico.org.mx}.
\newblock
  \bibinfo{howpublished}{\url{https://www.banxico.org.mx/SieInternet/consultarDirectorioInternetAction.do?sector=8&accion=consultarCuadro&idCuadro=CP150&locale=es}}.
\newblock
\newblock
\shownote{[Accessed 07-10-2025]}.


\bibitem[Bank(2019)]%
        {europaMoneyPrivate}
\bibfield{author}{\bibinfo{person}{European~Central Bank}.}
  \bibinfo{year}{2019}\natexlab{}.
\newblock \bibinfo{title}{{M}oney and private currencies: reflections on
  {L}ibra --- ecb.europa.eu}.
\newblock
  \bibinfo{howpublished}{\url{https://www.ecb.europa.eu/press/key/date/2019/html/ecb.sp190902~aedded9219.en.html}}.
\newblock
\newblock
\shownote{[Accessed 26-09-2025]}.


\bibitem[Baum{\"o}hl and V{\`y}rost(2011)]%
        {baumohl2011stock}
\bibfield{author}{\bibinfo{person}{Eduard Baum{\"o}hl} {and}
  \bibinfo{person}{Tom{\'a}{\v{s}} V{\`y}rost}.}
  \bibinfo{year}{2011}\natexlab{}.
\newblock \showarticletitle{Stock Market Integration: Granger Causality Testing
  with Respect to Nonsynchronous Trading Effects.}
\newblock \bibinfo{journal}{\emph{Finance a Uver: Czech Journal of Economics \&
  Finance}} \bibinfo{volume}{61}, \bibinfo{number}{1} (\bibinfo{year}{2011}).
\newblock


\bibitem[Ben-David et~al\mbox{.}(2020)]%
        {ben2020performance}
\bibfield{author}{\bibinfo{person}{Itzhak Ben-David}, \bibinfo{person}{Justin
  Birru}, {and} \bibinfo{person}{Andrea Rossi}.}
  \bibinfo{year}{2020}\natexlab{}.
\newblock \bibinfo{booktitle}{\emph{The performance of hedge fund performance
  fees}}.
\newblock \bibinfo{type}{{T}echnical {R}eport}. \bibinfo{institution}{National
  Bureau of Economic Research}.
\newblock


\bibitem[{BoE}(2026)]%
        {bankofenglandBankEngland}
\bibfield{author}{\bibinfo{person}{{BoE}}.} \bibinfo{year}{2026}\natexlab{}.
\newblock \bibinfo{title}{{B}ank of {E}ngland {D}atabase ---
  wwwtest.bankofengland.co.uk}.
\newblock
  \bibinfo{howpublished}{\url{https://wwwtest.bankofengland.co.uk/boeapps/database/default.asp}}.
\newblock
\newblock
\shownote{[Accessed 26-01-2026]}.


\bibitem[Bullmann et~al\mbox{.}(2019)]%
        {bullmann2019search}
\bibfield{author}{\bibinfo{person}{Dirk Bullmann}, \bibinfo{person}{Jonas
  Klemm}, {and} \bibinfo{person}{Andrea Pinna}.}
  \bibinfo{year}{2019}\natexlab{}.
\newblock \bibinfo{booktitle}{\emph{In search for stability in crypto-assets:
  are stablecoins the solution?}}
\newblock Number 230. \bibinfo{publisher}{ECB Occasional Paper}.
\newblock


\bibitem[Chen et~al\mbox{.}(1986)]%
        {chen1986economic}
\bibfield{author}{\bibinfo{person}{Nai-Fu Chen}, \bibinfo{person}{Richard
  Roll}, {and} \bibinfo{person}{Stephen~A Ross}.}
  \bibinfo{year}{1986}\natexlab{}.
\newblock \showarticletitle{Economic forces and the stock market}.
\newblock \bibinfo{journal}{\emph{Journal of business}} (\bibinfo{year}{1986}),
  \bibinfo{pages}{383--403}.
\newblock


\bibitem[{Circle}(2018)]%
        {usdcUSDCWorlds}
\bibfield{author}{\bibinfo{person}{{Circle}}.} \bibinfo{year}{2018}\natexlab{}.
\newblock \bibinfo{title}{{U}{S}{D}{C} | {T}he world’s largest regulated
  digital dollar --- usdc.com}.
\newblock \bibinfo{howpublished}{\url{https://www.usdc.com/}}.
\newblock
\newblock
\shownote{[Accessed 10-09-2025]}.


\bibitem[{Diem Association}(2020)]%
        {diemDiemBlockchain}
\bibfield{author}{\bibinfo{person}{{Diem Association}}.}
  \bibinfo{year}{2020}\natexlab{}.
\newblock \bibinfo{title}{{T}he {D}iem {B}lockchain | {D}iem {D}ocumentation
  --- developers.diem.com}.
\newblock
  \bibinfo{howpublished}{\url{https://developers.diem.com/docs/technical-papers/the-diem-blockchain-paper/}}.
\newblock
\newblock
\shownote{[Accessed 17-09-2025]}.


\bibitem[Elton et~al\mbox{.}(2003)]%
        {elton2003incentive}
\bibfield{author}{\bibinfo{person}{Edwin~J Elton}, \bibinfo{person}{Martin~J
  Gruber}, {and} \bibinfo{person}{Christopher~R Blake}.}
  \bibinfo{year}{2003}\natexlab{}.
\newblock \showarticletitle{Incentive fees and mutual funds}.
\newblock \bibinfo{journal}{\emph{The Journal of Finance}}
  \bibinfo{volume}{58}, \bibinfo{number}{2} (\bibinfo{year}{2003}),
  \bibinfo{pages}{779--804}.
\newblock


\bibitem[Emmett et~al\mbox{.}(2023)]%
        {emmett2023flatcoins}
\bibfield{author}{\bibinfo{person}{Jeff Emmett}, \bibinfo{person}{Danilo~Lessa
  Bernardineli}, {and} \bibinfo{person}{Jamsheed Shorish}.}
  \bibinfo{year}{2023}\natexlab{}.
\newblock \showarticletitle{Flatcoins: Inflation-Adjusted Stablecoins}.
\newblock  (\bibinfo{year}{2023}).
\newblock


\bibitem[Escobar-Anel et~al\mbox{.}(2018)]%
        {escobar2018optimal}
\bibfield{author}{\bibinfo{person}{Marcos Escobar-Anel},
  \bibinfo{person}{Vincent H{\"o}hn}, \bibinfo{person}{Luis Seco}, {and}
  \bibinfo{person}{Rudi Zagst}.} \bibinfo{year}{2018}\natexlab{}.
\newblock \showarticletitle{Optimal fee structures in hedge funds}.
\newblock \bibinfo{journal}{\emph{Journal of Asset Management}}
  \bibinfo{volume}{19}, \bibinfo{number}{7} (\bibinfo{year}{2018}),
  \bibinfo{pages}{522--542}.
\newblock


\bibitem[Eun and Shim(1989)]%
        {eun1989international}
\bibfield{author}{\bibinfo{person}{Cheol~S Eun} {and} \bibinfo{person}{Sangdal
  Shim}.} \bibinfo{year}{1989}\natexlab{}.
\newblock \showarticletitle{International transmission of stock market
  movements}.
\newblock \bibinfo{journal}{\emph{Journal of financial and quantitative
  Analysis}} \bibinfo{volume}{24}, \bibinfo{number}{2} (\bibinfo{year}{1989}),
  \bibinfo{pages}{241--256}.
\newblock


\bibitem[Fama and French(1992)]%
        {fama1992cross}
\bibfield{author}{\bibinfo{person}{Eugene~F Fama} {and}
  \bibinfo{person}{Kenneth~R French}.} \bibinfo{year}{1992}\natexlab{}.
\newblock \showarticletitle{The cross-section of expected stock returns}.
\newblock \bibinfo{journal}{\emph{the Journal of Finance}}
  \bibinfo{volume}{47}, \bibinfo{number}{2} (\bibinfo{year}{1992}),
  \bibinfo{pages}{427--465}.
\newblock


\bibitem[Fisher(1913)]%
        {fisher1913compensated}
\bibfield{author}{\bibinfo{person}{Irving Fisher}.}
  \bibinfo{year}{1913}\natexlab{}.
\newblock \showarticletitle{A compensated dollar}.
\newblock \bibinfo{journal}{\emph{The Quarterly Journal of Economics}}
  \bibinfo{volume}{27}, \bibinfo{number}{2} (\bibinfo{year}{1913}),
  \bibinfo{pages}{213--235}.
\newblock


\bibitem[{Frax Finance}(2024)]%
        {fraxOverviewCPI}
\bibfield{author}{\bibinfo{person}{{Frax Finance}}.}
  \bibinfo{year}{2024}\natexlab{}.
\newblock \bibinfo{title}{{O}verview ({C}{P}{I} {P}eg \& {M}echanics) | {F}rax
  {F}inance ¤ --- docs.frax.finance}.
\newblock
  \bibinfo{howpublished}{\url{https://docs.frax.finance/frax-price-index/overview-cpi-peg-and-mechanics}}.
\newblock
\newblock
\shownote{[Accessed 17-09-2025]}.


\bibitem[Grobys et~al\mbox{.}(2025)]%
        {grobys2025stablecoin}
\bibfield{author}{\bibinfo{person}{Klaus Grobys}, \bibinfo{person}{Juha-Pekka
  Junttila}, {and} \bibinfo{person}{James~W Kolari}.}
  \bibinfo{year}{2025}\natexlab{}.
\newblock \showarticletitle{A Stablecoin That’s Actually Stable: A Portfolio
  Optimization Approach}.
\newblock \bibinfo{journal}{\emph{Journal of Financial Stability}}
  (\bibinfo{year}{2025}), \bibinfo{pages}{101458}.
\newblock


\bibitem[Hauksson et~al\mbox{.}(2001)]%
        {hauksson2001multivariate}
\bibfield{author}{\bibinfo{person}{HA Hauksson}, \bibinfo{person}{M Dacorogna},
  \bibinfo{person}{T Domenig}, \bibinfo{person}{U Mller}, {and}
  \bibinfo{person}{G Samorodnitsky}.} \bibinfo{year}{2001}\natexlab{}.
\newblock \showarticletitle{Multivariate extremes, aggregation and risk
  estimation}.
\newblock \bibinfo{journal}{\emph{Quantitative Finance}} \bibinfo{volume}{1},
  \bibinfo{number}{1} (\bibinfo{year}{2001}), \bibinfo{pages}{79--95}.
\newblock


\bibitem[Hayashi and Yoshida(2005)]%
        {hayashi2005covariance}
\bibfield{author}{\bibinfo{person}{Takaki Hayashi} {and}
  \bibinfo{person}{Nakahiro Yoshida}.} \bibinfo{year}{2005}\natexlab{}.
\newblock \showarticletitle{On covariance estimation of non-synchronously
  observed diffusion processes}.
\newblock \bibinfo{journal}{\emph{Bernoulli}} \bibinfo{volume}{11},
  \bibinfo{number}{2} (\bibinfo{year}{2005}), \bibinfo{pages}{359--379}.
\newblock


\bibitem[Ho(2012)]%
        {ho2012globalization}
\bibfield{author}{\bibinfo{person}{Lok~Sang Ho}.}
  \bibinfo{year}{2012}\natexlab{}.
\newblock \showarticletitle{Globalization, exports, and effective exchange rate
  indices}.
\newblock \bibinfo{journal}{\emph{Journal of International Money and Finance}}
  \bibinfo{volume}{31}, \bibinfo{number}{5} (\bibinfo{year}{2012}),
  \bibinfo{pages}{996--1007}.
\newblock


\bibitem[Ho(2018)]%
        {ho2018search}
\bibfield{author}{\bibinfo{person}{Lok~Sang Ho}.}
  \bibinfo{year}{2018}\natexlab{}.
\newblock \showarticletitle{In search of a unit of stable global purchasing
  power}.
\newblock \bibinfo{journal}{\emph{International Review of Economics \&
  Finance}}  \bibinfo{volume}{56} (\bibinfo{year}{2018}),
  \bibinfo{pages}{99--108}.
\newblock


\bibitem[Hovanov et~al\mbox{.}(2004)]%
        {hovanov2004computing}
\bibfield{author}{\bibinfo{person}{Nikolai~V Hovanov}, \bibinfo{person}{James~W
  Kolari}, {and} \bibinfo{person}{Mikhail~V Sokolov}.}
  \bibinfo{year}{2004}\natexlab{}.
\newblock \showarticletitle{Computing currency invariant indices with an
  application to minimum variance currency baskets}.
\newblock \bibinfo{journal}{\emph{Journal of Economic Dynamics and Control}}
  \bibinfo{volume}{28}, \bibinfo{number}{8} (\bibinfo{year}{2004}),
  \bibinfo{pages}{1481--1504}.
\newblock


\bibitem[{IMF}(2023)]%
        {imfWhatSDR}
\bibfield{author}{\bibinfo{person}{{IMF}}.} \bibinfo{year}{2023}\natexlab{}.
\newblock \bibinfo{title}{{W}hat is the {S}{D}{R}? --- imf.org}.
\newblock
  \bibinfo{howpublished}{\url{https://www.imf.org/en/About/Factsheets/Sheets/2023/special-drawing-rights-sdr}}.
\newblock
\newblock
\shownote{[Accessed 18-09-2025]}.


\bibitem[{IMF}(2026)]%
        {imfDataHome}
\bibfield{author}{\bibinfo{person}{{IMF}}.} \bibinfo{year}{2026}\natexlab{}.
\newblock \bibinfo{title}{{D}ata {H}ome --- data.imf.org}.
\newblock \bibinfo{howpublished}{\url{https://data.imf.org/en}}.
\newblock
\newblock
\shownote{[Accessed 26-01-2026]}.


\bibitem[Ito(2024)]%
        {kensuke2024cryptoeconomics}
\bibfield{author}{\bibinfo{person}{Kensuke Ito}.}
  \bibinfo{year}{2024}\natexlab{}.
\newblock \showarticletitle{Cryptoeconomics and tokenomics as economics: A
  survey with opinions}. In \bibinfo{booktitle}{\emph{2024 IEEE International
  Conference on Blockchain and Cryptocurrency (ICBC)}}. IEEE,
  \bibinfo{pages}{729--746}.
\newblock


\bibitem[Ito et~al\mbox{.}(2020)]%
        {ito2020stablecoin}
\bibfield{author}{\bibinfo{person}{Kensuke Ito}, \bibinfo{person}{Makiko Mita},
  \bibinfo{person}{Shohei Ohsawa}, {and} \bibinfo{person}{Hideyuki Tanaka}.}
  \bibinfo{year}{2020}\natexlab{}.
\newblock \showarticletitle{What is stablecoin?: A survey on its mechanism and
  potential as decentralized payment systems}.
\newblock \bibinfo{journal}{\emph{International Journal of Service and
  Knowledge Management}} \bibinfo{volume}{4}, \bibinfo{number}{2}
  (\bibinfo{year}{2020}), \bibinfo{pages}{71--86}.
\newblock


\bibitem[Kereiakes et~al\mbox{.}(2019)]%
        {kereiakes2019terra}
\bibfield{author}{\bibinfo{person}{Evan Kereiakes}, \bibinfo{person}{Marco
  Di~Maggio Do~Kwon}, {and} \bibinfo{person}{Nicholas Platias}.}
  \bibinfo{year}{2019}\natexlab{}.
\newblock \showarticletitle{Terra money: Stability and adoption}.
\newblock \bibinfo{journal}{\emph{White Paper, Apr}} (\bibinfo{year}{2019}).
\newblock


\bibitem[Ling et~al\mbox{.}(2025)]%
        {ling2025sok}
\bibfield{author}{\bibinfo{person}{Shengchen Ling}, \bibinfo{person}{Yuefeng
  Du}, \bibinfo{person}{Yajin Zhou}, \bibinfo{person}{Lei Wu},
  \bibinfo{person}{Cong Wang}, \bibinfo{person}{Xiaohua Jia}, {and}
  \bibinfo{person}{Houmin Yan}.} \bibinfo{year}{2025}\natexlab{}.
\newblock \showarticletitle{SoK: Stablecoin Designs, Risks, and the Stablecoin
  LEGO}.
\newblock \bibinfo{journal}{\emph{arXiv preprint arXiv:2506.17622}}
  (\bibinfo{year}{2025}).
\newblock


\bibitem[Longin and Solnik(2001)]%
        {longin2001extreme}
\bibfield{author}{\bibinfo{person}{Francois Longin} {and}
  \bibinfo{person}{Bruno Solnik}.} \bibinfo{year}{2001}\natexlab{}.
\newblock \showarticletitle{Extreme correlation of international equity
  markets}.
\newblock \bibinfo{journal}{\emph{The journal of finance}}
  \bibinfo{volume}{56}, \bibinfo{number}{2} (\bibinfo{year}{2001}),
  \bibinfo{pages}{649--676}.
\newblock


\bibitem[Louw(1987)]%
        {louw1987ecu}
\bibfield{author}{\bibinfo{person}{Andr{\'e} Louw}.}
  \bibinfo{year}{1987}\natexlab{}.
\newblock \showarticletitle{The ECU: facts and prospects}.
\newblock \bibinfo{journal}{\emph{Revue de l’euro}} \bibinfo{number}{51}
  (\bibinfo{year}{1987}).
\newblock


\bibitem[{MakerDAO}(2022)]%
        {dai2022maker}
\bibfield{author}{\bibinfo{person}{{MakerDAO}}.}
  \bibinfo{year}{2022}\natexlab{}.
\newblock \bibinfo{title}{The Maker Protocol: MakerDAO's Multi-Collateral Dai
  (MCD) System}.
\newblock
\urldef\tempurl%
\url{https://makerdao.com/whitepaper/White%20Paper%20-The%20Maker%20Protocol_%20MakerDAO's%20Multi-Collateral%20Dai%20(MCD)%20System-FINAL-%20021720.pdf}
\showURL{%
\tempurl}
\newblock
\shownote{[Accessed 11-09-2025]}.


\bibitem[Makin(1978)]%
        {makin1978portfolio}
\bibfield{author}{\bibinfo{person}{John~H Makin}.}
  \bibinfo{year}{1978}\natexlab{}.
\newblock \showarticletitle{Portfolio theory and the problem of foreign
  exchange risk}.
\newblock \bibinfo{journal}{\emph{The Journal of Finance}}
  \bibinfo{volume}{33}, \bibinfo{number}{2} (\bibinfo{year}{1978}),
  \bibinfo{pages}{517--534}.
\newblock


\bibitem[Manaster(1979)]%
        {manaster1979real}
\bibfield{author}{\bibinfo{person}{Steven Manaster}.}
  \bibinfo{year}{1979}\natexlab{}.
\newblock \showarticletitle{Real and nominal efficient sets}.
\newblock \bibinfo{journal}{\emph{The Journal of Finance}}
  \bibinfo{volume}{34}, \bibinfo{number}{1} (\bibinfo{year}{1979}),
  \bibinfo{pages}{93--102}.
\newblock


\bibitem[Markowits(1952)]%
        {markowits1952portfolio}
\bibfield{author}{\bibinfo{person}{Harry~M Markowits}.}
  \bibinfo{year}{1952}\natexlab{}.
\newblock \showarticletitle{Portfolio selection}.
\newblock \bibinfo{journal}{\emph{Journal of finance}} \bibinfo{volume}{7},
  \bibinfo{number}{1} (\bibinfo{year}{1952}), \bibinfo{pages}{71--91}.
\newblock


\bibitem[Nakamoto(2008)]%
        {nakamoto2008bitcoin}
\bibfield{author}{\bibinfo{person}{Satoshi Nakamoto}.}
  \bibinfo{year}{2008}\natexlab{}.
\newblock \showarticletitle{Bitcoin: A peer-to-peer electronic cash system}.
\newblock \bibinfo{journal}{\emph{Decentralized Business Review}}
  (\bibinfo{year}{2008}), \bibinfo{pages}{21260}.
\newblock


\bibitem[Nelsen(2006)]%
        {nelsen2006introduction}
\bibfield{author}{\bibinfo{person}{Roger~B Nelsen}.}
  \bibinfo{year}{2006}\natexlab{}.
\newblock \bibinfo{booktitle}{\emph{An introduction to copulas}}.
\newblock \bibinfo{publisher}{Springer}.
\newblock


\bibitem[Sharpe(1964)]%
        {sharpe1964capital}
\bibfield{author}{\bibinfo{person}{William~F Sharpe}.}
  \bibinfo{year}{1964}\natexlab{}.
\newblock \showarticletitle{Capital asset prices: A theory of market
  equilibrium under conditions of risk}.
\newblock \bibinfo{journal}{\emph{The journal of finance}}
  \bibinfo{volume}{19}, \bibinfo{number}{3} (\bibinfo{year}{1964}),
  \bibinfo{pages}{425--442}.
\newblock


\bibitem[Shiller(1998)]%
        {shiller1998indexed}
\bibfield{author}{\bibinfo{person}{Robert~J Shiller}.}
  \bibinfo{year}{1998}\natexlab{}.
\newblock \bibinfo{title}{Indexed units of account: theory and assessment of
  historical experience}.
\newblock


\bibitem[Shiller et~al\mbox{.}(1981)]%
        {shiller1981stock}
\bibfield{author}{\bibinfo{person}{Robert~J Shiller} {et~al\mbox{.}}}
  \bibinfo{year}{1981}\natexlab{}.
\newblock \showarticletitle{Do stock prices move too much to be justified by
  subsequent changes in dividends?}
\newblock  (\bibinfo{year}{1981}).
\newblock


\bibitem[Simmons(2021)]%
        {simmons2021regulating}
\bibfield{author}{\bibinfo{person}{Amanda Simmons}.}
  \bibinfo{year}{2021}\natexlab{}.
\newblock \showarticletitle{Regulating Libra: Will Legal and Regulatory
  Uncertainty Prevent the Launch of Facebook's Cryptocurrency Project?}
\newblock \bibinfo{journal}{\emph{J. Bus. \& Tech. L.}}  \bibinfo{volume}{16}
  (\bibinfo{year}{2021}), \bibinfo{pages}{83}.
\newblock


\bibitem[Solnik(1978)]%
        {solnik1978inflation}
\bibfield{author}{\bibinfo{person}{Bruno~H Solnik}.}
  \bibinfo{year}{1978}\natexlab{}.
\newblock \showarticletitle{Inflation and optimal portfolio choices}.
\newblock \bibinfo{journal}{\emph{Journal of Financial and quantitative
  analysis}} \bibinfo{volume}{13}, \bibinfo{number}{5} (\bibinfo{year}{1978}),
  \bibinfo{pages}{903--925}.
\newblock


\bibitem[{Tether}(2015)]%
        {tetherTransparency}
\bibfield{author}{\bibinfo{person}{{Tether}}.} \bibinfo{year}{2015}\natexlab{}.
\newblock \bibinfo{title}{{T}ransparency --- tether.to}.
\newblock
  \bibinfo{howpublished}{\url{https://tether.to/en/transparency/\#usdt}}.
\newblock
\newblock
\shownote{[Accessed 10-09-2025]}.


\bibitem[{The Libra Association}(2019)]%
        {archiveLibraWhite}
\bibfield{author}{\bibinfo{person}{{The Libra Association}}.}
  \bibinfo{year}{2019}\natexlab{}.
\newblock \bibinfo{title}{{L}ibra {W}hite {P}aper : {T}he {L}ibra {A}ssociation
  : {F}ree {D}ownload, {B}orrow, and {S}treaming : {I}nternet {A}rchive ---
  archive.org}.
\newblock
  \bibinfo{howpublished}{\url{https://archive.org/details/facebooklibrawhitepaper}}.
\newblock
\newblock
\shownote{[Accessed 18-09-2025]}.


\bibitem[Tyler(1987)]%
        {tyler1987distribution}
\bibfield{author}{\bibinfo{person}{David~E Tyler}.}
  \bibinfo{year}{1987}\natexlab{}.
\newblock \showarticletitle{A distribution-free M-estimator of multivariate
  scatter}.
\newblock \bibinfo{journal}{\emph{The annals of Statistics}}
  (\bibinfo{year}{1987}), \bibinfo{pages}{234--251}.
\newblock


\bibitem[Woltzenlogel~Paleo(2023)]%
        {woltzenlogel2023stablecoin}
\bibfield{author}{\bibinfo{person}{Bruno Woltzenlogel~Paleo}.}
  \bibinfo{year}{2023}\natexlab{}.
\newblock \showarticletitle{Stablecoin}.
\newblock In \bibinfo{booktitle}{\emph{Encyclopedia of Cryptography, Security
  and Privacy}}. \bibinfo{publisher}{Springer}, \bibinfo{pages}{1--5}.
\newblock


\bibitem[Zetzsche et~al\mbox{.}(2019)]%
        {zetzsche2019regulating}
\bibfield{author}{\bibinfo{person}{Dirk~A Zetzsche}, \bibinfo{person}{Ross~P
  Buckley}, {and} \bibinfo{person}{Douglas~W Arner}.}
  \bibinfo{year}{2019}\natexlab{}.
\newblock \showarticletitle{Regulating LIBRA: The transformative potential of
  Facebook’s cryptocurrency and possible regulatory responses}.
\newblock  (\bibinfo{year}{2019}).
\newblock


\end{thebibliography}

\appendix

\section{Estimation of $\Sigma_t$}
\label{app-sigma}

Estimation of $\Sigma_t$ is based on observable price log-returns over the past $3N$ periods up to time $t$.
Let $\mathcal{W}_t \coloneqq \{t-3N+1, t-3N+2,\ldots, t\}$ denote this estimation window.
For notational convenience, define $\{\mathbf{\Delta p}_{i}\} \coloneqq \{\mathbf{\Delta p}_{i,s}\}_{s \in \mathcal{W}_t}$ as the sequence of observable price log-returns in units of $i$ over $\mathcal{W}_t$, and $\{\mathbf{\Delta \widehat{v}}\} \coloneqq \{\mathbf{\Delta \widehat{v}}_{s}\}_{s \in \mathcal{W}_t}$ as the corresponding sequence of MLVs.

\paragraph{Iterative procedure.}
Since $\mathbf{\Delta \widehat{v}}_{t}$ depends on $\Sigma_t$ through \eqref{eqn:mlv}, while $\Sigma_t$ in turn represents the covariance structure of $\mathbf{\Delta \widehat{v}}_{t}$, $\Sigma_t$ is estimated via the following iterative procedure.

\begin{itemize}
    \item \textbf{Step 1 (Initialization).}
    Set $\Sigma_t^{(k=0)}$ to the identity matrix.

    \item \textbf{Step 2 (MLV update).}
    Given $\Sigma_t^{(k)}$, compute the scalar adjustment $\Delta \widehat{v}_{i,s}^{(k)}$ for all $s \in \mathcal{W}_t$ from 
    $\{\mathbf{\Delta p}_{i}\}$ using \eqref{eqn:mlv}, and construct $\{\mathbf{\Delta \widehat{v}}\}^{(k)}$ using $\mathbf{\Delta \widehat{v}}_{s}^{(k)} = \mathbf{\Delta p}_{i,s} + \Delta \widehat{v}_{i,s}^{(k)} \mathbf{1}$.

    \item \textbf{Step 3 (Covariance update).}
    Update $\Sigma_t^{(k+1)}$ from the resulting
    $\{\mathbf{\Delta \widehat{v}}\}^{(k)}$
    using a robust covariance estimator (described below).

    \item \textbf{Step 4 (Iteration).} 
    Set $k \leftarrow k+1$ and repeat Steps 2–4 until convergence. 
\end{itemize}

\noindent
Starting from the identity matrix, this procedure estimates
$\Sigma_t$ from past observable price log-returns by iteratively reconstructing
$\mathbf{\Delta \widehat{v}}_{t}$, without requiring MLV-denominated data ex ante.

\paragraph{Robust covariance estimation.}
In Step 3 above, we estimate the covariance \emph{structure} of $\mathbf{\Delta \widehat{v}}_{t}$ using \textit{Tyler’s M-estimator} \cite{tyler1987distribution}, which is robust to heavy-tailed financial returns.
Tyler’s estimator consistently identifies the covariance structure under elliptical distributions up to a positive scale factor, which is inconsequential here because the MLV is invariant to positive scalar rescaling of $\Sigma_t$\footnote{This invariance arises because any positive scalar rescaling of $\Sigma_t$
cancels out in \eqref{eqn:mlv}.}.
Formally, Tyler’s M-estimator $\hat{\Sigma}_t$ is defined (up to scale) as the solution to the fixed-point equation:

\begin{equation}
\hat{\Sigma}_t
=
\frac{N}{|\mathcal{W}_t|}
\sum_{s \in \mathcal{W}_t}
\frac{\mathbf{\Delta \widehat{v}}_{s}\mathbf{\Delta \widehat{v}}_{s}^{\top}}
{\mathbf{\Delta \widehat{v}}_{s}^{\top}\hat{\Sigma}_t^{-1}\mathbf{\Delta \widehat{v}}_{s}},
\label{eq:tyler}
\end{equation}

\noindent
where $|\mathcal{W}_t|$ denotes the sample size within the estimation window, i.e., $3N$.
The fixed-point equation \eqref{eq:tyler} is solved numerically by iterating the above procedure until convergence.

\section{\texorpdfstring{Definition of $\widehat{\boldsymbol{\mu}}_{i, t}$}{Definition of mu-hat_t}} \label{app-mu}
Let $d_{ij,t}$ denote the dividend of asset $j$ distributed at that time $t$, expressed in unit of $i$.
In this case, its dividend return is defined as:

\begin{equation}
r^{\mathrm{div}}_{ij,t}
\coloneqq
\frac{d_{ij,t}}{p_{ij,t-1}}.
\end{equation}

\noindent
Collecting these returns across all assets yields the dividend return vector:
\begin{equation}
\mathbf{r}^{\mathrm{div}}_{i,t}
\coloneqq
\bigl(
r^{\mathrm{div}}_{i1,t},
\dots,
r^{\mathrm{div}}_{iN,t}
\bigr)^{\top}.
\end{equation}

\noindent
At each rebalancing time $t$, the expected dividend return vector is estimated by using historical data over a fixed lookback window.
Specifically, we define:
\begin{equation}
\widehat{\boldsymbol{\mu}}_{i, t}
\coloneqq
\frac{1}{\lvert \mathcal{L}_t \rvert}
\sum_{s \in \mathcal{L}_t}
\mathbf{r}^{\mathrm{div}}_{i,s},
\end{equation}
where $\mathcal{L}_t$ denotes the set of past trading times within the lookback period.
Here, the vector $\widehat{\boldsymbol{\mu}}_{i,t}$ provides an empirical estimate of the (unobserved) true expected dividend return vector $\boldsymbol{\mu}_{i,t}$.

\end{document}